\documentclass[10pt, a4paper]{article}%
\usepackage{amsmath,amssymb,eucal}
\usepackage{amsfonts}
\usepackage{mathrsfs}
\usepackage{slashed}
\usepackage{graphicx}
\usepackage{ifpdf}
 \usepackage{url}
\numberwithin{equation}{section} 
\def\beq{\begin{equation}}
\def\eeq{\end{equation}}

\def\bR{ {{\mathbb{R}}}}

\def\Tr{ {{\rm{Tr}}} }

\newcommand{\sgn}{{\rm sgn}}

\newcommand{\pk}[1]{p_{\kappa}}

\newcommand{\ol}{\overline{l}}
\newcommand{\ul}{\underline{l}}
\newcommand{\oL}{\overline{L}}
\newcommand{\uL}{\underline{L}}
\newcommand{\on}{\overline{n}}
\newcommand{\un}{\underline{n}}

\newtheorem{defn}{{\bf Definition}}[section]
\newtheorem{thm}[defn]{{\bf Theorem}}
\newtheorem{cor}[defn]{{\bf Corollary}}
\newtheorem{lem}[defn]{{\bf Lemma}}

\newtheorem{rem}[defn]{{\bf Remark}}

\newtheorem{notation}[defn]{Notation}
\newenvironment{proof}[1][Proof]{\textbf{#1.} }{\hfill \rule{0.5em}{0.5em}}
\begin{document}

\title{Quantized Curvature in Loop Quantum Gravity}
\author{Adrian P. C. Lim \\
Email: ppcube@gmail.com
}

\date{}

\maketitle

\begin{abstract}
A hyperlink is a finite set of non-intersecting simple closed curves in $\mathbb{R} \times \mathbb{R}^3$. Let $S$ be an orientable surface in $\mathbb{R} \times \mathbb{R}^3$. The Einstein-Hilbert action $S(e,\omega)$ is defined on the vierbein $e$ and a $\mathfrak{su}(2)\times\mathfrak{su}(2)$-valued connection $\omega$, which are the dynamical variables in General Relativity. Define a functional $F_S(\omega)$, by integrating the curvature $d\omega + \omega \wedge \omega$ over the surface $S$, which is $\mathfrak{su}(2)\times\mathfrak{su}(2)$-valued. We integrate $F_S(\omega)$ against a holonomy operator of a hyperlink $L$, disjoint from $S$, and the exponential of the Einstein-Hilbert action, over the space of vierbeins $e$ and $\mathfrak{su}(2)\times\mathfrak{su}(2)$-valued connections $\omega$. Using our earlier work done on Chern-Simons path integrals in $\mathbb{R}^3$, we will write this infinite dimensional path integral as the limit of a sequence of Chern-Simons integrals. Our main result shows that the quantized curvature can be computed from the linking number between $L$ and $S$.
\end{abstract}

\hspace{.35cm}{\small {\bf MSC} 2010: } 83C45, 81S40, 81T45, 57R56 \\
\indent \hspace{.35cm}{\small {\bf Keywords}: Curvature, Path integral, Einstein-Hilbert, Loop representation, Quantum gravity}




\section{Quantization of curvature}\label{s.intro}

In the quantization of area and volume as done in \cite{rovelli1995discreteness}, they used canonical quantization on the dynamical Ashtekar variables to derive area and volume operators. Such a method will no longer work for curvature, as the curvature involves derivatives of the dynamical Ashtekar variables.

One way to overcome this is to quantize the holonomy of a connection. This is done by expressing the curvature as the limit of the holonomy around a loop, divided by the area enclosed by the loop, as the area shrinks to zero. The details can be found in \cite{PhysRevLett.96.141301}.

However, as shown in \cite{EH-Lim03}, in loop quantum gravity, the area should be interpreted as a quantized operator, possibly taking the value 0 in its range. Thus, one will run into serious problems in trying to define the reciprocal of an area operator. The authors in \cite{PhysRevD.73.124038} also knowledged that the limit of the above procedure do not exist, as the eigenvalues of the area operator is discrete.

Hence, we will use a completely different approach to quantize curvature, via path integrals. By writing a consistent set of rules, one can make sense and define a path integral, that involves an area, volume and curvature functional. See \cite{EH-Lim02}. We will now summarize this path integral quantization of curvature.

Consider $\bR^4 \equiv \bR \times \bR^3$, whereby $\bR$ will be referred to as the time-axis and $\bR^3$ is the spatial 3-dimensional Euclidean space. Fix the standard coordinates on $\bR^4\equiv \bR \times \bR^3 $, with time coordinate $x_0$ and spatial coordinates $(x_1, x_2, x_3)$.

Let $\{e_i\}_{i=1}^3$ be the standard basis in $\bR^3$. And $\Sigma_i$ is the plane in $\bR^3$, containing the origin, whose normal is given by $e_i$. Let $\pi_0: \bR^4 \rightarrow \bR^3$ denote this projection and for $i=1, 2, 3$, $\pi_i: \bR^4 \rightarrow \bR \times \Sigma_i$ denote a projection.

Let $\tilde{V} \rightarrow \bR \times \bR^3$ be a 4-dimensional vector bundle, with structure group $SO(3,1)$. This implies that $\tilde{V}$ is endowed with a metric, $\eta^{ab}$, of signature $(-, +, +, +)$, and a volume form $\epsilon_{abcd}$.

Suppose $\tilde{V}$ has the same topological type as $T\bR^4$, so that isomorphisms between $\tilde{V}$ and $T\bR^4$ exist. Hence we may assume that $\tilde{V}$ is a trivial bundle over $\bR^4$. That is, $\tilde{V} \equiv \bR^4 \times V \rightarrow \bR^4$ will be our trivial bundle in consideration, for some fixed vector space $V$ of dimension 4.

However, there is no natural choice of an isomorphism. A vierbein $e$ is a choice of isomorphism between $T\bR^4$ and $\tilde{V}$. It may be regarded as a $\tilde{V}$-valued one form, obeying a certain condition of invertibility. A spin connection $\omega_{\mu \gamma}^a$ on $\tilde{V}$, is anti-symmetric in its indices $\mu$, $\gamma$. It takes values in $\Lambda^2 (\tilde{V})$, whereby $\Lambda^k (\tilde{V})$ denotes the $k$-th antisymmetric tensor power or exterior power of $\tilde{V}$. The isomorphism $e$ and the connection $w$ can be regarded as the dynamical variables of General Relativity.

\begin{notation}
Indices such as $a,b,c, d$ and greek indices such as $\mu,\gamma, \alpha, \beta$ will take values from 0 to 3. We will use the greek indices to index the basis in $\Lambda^2(V)$.

Indices labeled $i, j, k$,  will only take values from 1 to 3. These indices will keep track of the spatial coordinate $x_i$.
\end{notation}

The curvature tensor is then defined as
\beq R_{\mu\gamma}^{ab} = \partial_a \omega_{\mu \gamma}^{b} - \partial_b \omega_{\mu \gamma}^{a} + [\omega^a, \omega^b]_{\mu\gamma},\ \partial_a \equiv \partial/\partial x_a, \label{e.r.1} \eeq
or as $R = d\omega + \omega \wedge \omega$. It can be regarded as a two form with values in $\Lambda^2 (\tilde{V})$.

Throughout this article, we adopt Einstein's summation convention, that is, we sum over repeated superscripts and subscripts, unless stated otherwise. Using the above notations, the Einstein-Hilbert action is written as \beq S_{EH}(e, \omega) = \frac{1}{8}\int_{\bR^4}\epsilon^{\mu \gamma \alpha \beta}\epsilon_{abcd}\ e_\mu^{a}e_\gamma^b R_{\alpha\beta}^{cd}. \label{e.eh.2} \eeq Here, $\epsilon^{\mu \gamma \alpha \beta} \equiv \epsilon_{\mu \gamma \alpha \beta}$ is equal to 1 if the number of transpositions required to permute $(0123)$ to $(\mu\gamma\alpha\beta)$ is even; otherwise it takes the value -1.

The expression $e \wedge e \wedge R$ is a four form on $\bR^4$ taking values in $\tilde{V} \otimes \tilde{V} \otimes \Lambda^2(\tilde{V})$ which maps to $\Lambda^4(\tilde{V})$. But $\tilde{V}$ with the structure group $SO(3,1)$ has a natural volume form, so a section of $\Lambda^4(\tilde{V})$ may be canonically regarded as a function. Thus Equation (\ref{e.eh.2}) is an invariantly defined integral.
By varying Equation (\ref{e.eh.2}) with respect to $e$, we will obtain the Einstein equations in vacuum. See \cite{Witten:1988hc}.

The metric $\eta^{ab}$ on $\tilde{V}$, together with the isomorphism $e$ between $T\bR^4$ and $\tilde{V}$, gives a (non-degenerate) metric $g^{ab} = e_\mu^a e_\gamma^b \eta^{\mu\gamma}$ on $T\bR^4$. By varying Equation (\ref{e.eh.2}) with respect to the connection $\omega$, we will obtain an equation that identifies $\omega$ as the Levi-Civita connection associated with the metric $g^{ab}$.

Fix a basis $\{E^\gamma\}_{\gamma=0}^3$ in $V$. Write $E^{\gamma \mu}= E^\gamma \wedge E^\mu \in \Lambda^2(V)$, thus $\{E^{\gamma \mu}\}_{0\leq \gamma<\mu\leq 3}$ is a basis for $\Lambda^2(V)$. Let $\mathfrak{su}(2)$ be the Lie Algebra of $SU(2)$. Choose the following basis for $\mathfrak{su}(2)$, \beq \breve{e}_1 :=
\frac{1}{2}\left(
  \begin{array}{cc}
    0 &\ 1 \\
    -1 &\ 0 \\
  \end{array}
\right),\ \ \breve{e}_2 :=
\frac{1}{2}\left(
  \begin{array}{cc}
    0 &\ i \\
    i &\ 0 \\
  \end{array}
\right),\ \ \breve{e}_3 :=
\frac{1}{2}\left(
  \begin{array}{cc}
    i &\ 0 \\
    0 &\ -i \\
  \end{array}
\right). \nonumber \eeq Write
\begin{align*}
\hat{E}^{01} = (\breve{e}_1, 0),\ \ \hat{E}^{02} = (\breve{e}_2, 0), \ \ \hat{E}^{03} = (\breve{e}_3, 0), \\
\hat{E}^{23} = (0, \breve{e}_1),\ \ \hat{E}^{31} = ( 0, \breve{e}_2), \ \ \hat{E}^{12} = (0, \breve{e}_3),
\end{align*}
all inside $\mathfrak{su}(2) \times \mathfrak{su}(2)$ and also denote
\beq \hat{E}^{\tau(1)} = \hat{E}^{23}, \ \ \hat{E}^{\tau(2)} = \hat{E}^{31}, \ \ \hat{E}^{\tau(3)} = \hat{E}^{12}. \nonumber \eeq
We will also write $\hat{E}^{\alpha\beta} = -\hat{E}^{\beta\alpha}$.

For $A, B, C, D \in \mathfrak{su}(2)$, we define the Lie bracket on $\mathfrak{su}(2) \times \mathfrak{su}(2)$ as \beq [(A,B), (C,D)] = ([A,C], [B, D]) \in \mathfrak{su}(2) \times \mathfrak{su}(2). \nonumber \eeq

We can map $\Lambda^2(V)$ to the Lie Algebra $\mathfrak{su}(2) \times \mathfrak{su}(2)$ via a linear map $\xi$ as follows,
\beq E^{01} \mapsto (\breve{e}_1,0) \equiv \hat{E}^{01},\ E^{02} \mapsto (\breve{e}_2,0) \equiv \hat{E}^{02},\ E^{03} \mapsto (\breve{e}_3,0) \equiv \hat{E}^{03} \nonumber \eeq and \beq E^{23} \mapsto (0,\hat{e}_1) \equiv \hat{E}^{23},\ E^{31} \mapsto (0,\hat{e}_2) \equiv \hat{E}^{31},\ E^{12} \mapsto (0,\hat{e}_3) \equiv \hat{E}^{12}. \nonumber \eeq 

Let $\Lambda^q(\bR^4)$ denote the q-th exterior power of $\bR^4$ and using the standard coordinates on $\bR^4$, we choose the canonical basis $\{dx_0, dx_1, dx_2, dx_3\}$ for $\Lambda^1(\bR^4)$. And we choose $\{dx_1, dx_2, dx_3\}$ to be a basis for the subspace  $\Lambda^1(\bR^3)$ in $\Lambda^1(\bR^4)$. A basis for $\Lambda^2(\bR^4)$ is hence given by \beq \{dx_0 \wedge dx_1, dx_0 \wedge dx_2, dx_0 \wedge dx_3, dx_2\wedge dx_3, dx_3 \wedge dx_1, dx_1 \wedge dx_2\}. \nonumber \eeq

Let $\overline{\mathcal{S}}_\kappa(\bR^4) \subset L^2(\bR^4)$ be a Schwartz space, as defined in \cite{EH-Lim02} and define
\begin{align*}
L_\omega :=& \overline{\mathcal{S}}_\kappa(\bR^4) \otimes \Lambda^1(\bR^3)\otimes \mathfrak{su}(2) \times \mathfrak{su}(2), \\
L_e :=& \overline{\mathcal{S}}_\kappa(\bR^4) \otimes \Lambda^1(\bR^3)\otimes V.
\end{align*}

The vierbein $e$ described earlier can be written as \beq e = B^i_\gamma \otimes dx_i\otimes E^\gamma \in \overline{\mathcal{S}}_\kappa(\bR^4) \otimes \Lambda^1(\bR^3) \otimes V = L_e. \nonumber \eeq

Instead of writing the spin connection $\omega$ to be $\Lambda^2(V)$-valued, we will write it to be $\mathfrak{su}(2) \times \mathfrak{su}(2)$-valued, i.e.
\begin{align*}
\omega =& A^i_{\alpha\beta} \otimes dx_i\otimes \hat{E}^{\alpha\beta} \in \overline{\mathcal{S}}_\kappa(\bR^4) \otimes \Lambda^1(\bR^3)\otimes \mathfrak{su}(2) \times \mathfrak{su}(2)= L_\omega.
\end{align*}
There is an implied sum over repeated indices. This $\omega$ is now a $\mathfrak{su}(2) \times \mathfrak{su}(2)$-valued one form. Its curvature form $R = d\omega + \omega \wedge \omega$ is now a $\mathfrak{su}(2) \times \mathfrak{su}(2)$-valued two form.

\begin{rem}
\begin{enumerate}
  \item Note that $A^{i}_{\alpha\beta} = -A^{i}_{\beta\alpha}\in \overline{\mathcal{S}}_\kappa(\bR^4)$.
  \item The authors in \cite{PhysRevD.73.124038} also considered the spin connection to be $\mathfrak{su}(2)$-valued.
  \item Even though the above curvature form $R$ is $\mathfrak{su}(2) \times \mathfrak{su}(2)$-valued, do note that its components $\{R_{\mu\gamma}^{ab}\}$ are indeed the components of the curvature of the vector bundle $\tilde{V} \rightarrow \bR^4$. Since $\tilde{V}$ is isomorphic to $T\bR^4$, therefore quantizing $R$ is equivalent to quantizing the curvature of $T\bR^4$.
\end{enumerate}
\end{rem}

In \cite{EH-Lim02}, after applying axial gauge fixing, as ($\partial_0 \equiv \partial/\partial x_0$), the Einstein-Hilbert action defined in Equation (\ref{e.eh.2}) now becomes
\begin{align*}
S_{EH}(e, \omega) :=
\frac{1}{8}&\int_{\bR^4}\epsilon^{abcd}B^1_\gamma B^2_\mu[\hat{E}^{\gamma \mu}]_{ab} \cdot \partial_0 A^3_{\alpha\beta}[\hat{E}^{\alpha\beta}]_{cd} dx_1\wedge dx_2 \wedge dx_0 \wedge dx_3\\
+& \frac{1}{8}\int_{\bR^4}\epsilon^{abcd}B^2_\gamma B^3_\mu[\hat{E}^{\gamma \mu}]_{ab} \cdot \partial_0 A^1_{\alpha\beta}[\hat{E}^{\alpha\beta}]_{cd} dx_2\wedge dx_3 \wedge dx_0 \wedge dx_1\\
+&\frac{1}{8}\int_{\bR^4}\epsilon^{abcd}B^3_\gamma B^1_\mu[\hat{E}^{\gamma \mu}]_{ab} \cdot \partial_0 A^2_{\alpha\beta}[\hat{E}^{\alpha\beta}]_{cd} dx_3\wedge dx_1 \wedge dx_0 \wedge dx_2.
\end{align*}

\begin{rem}
Note that a similar gauge fixing was also carried out in the analysis in \cite{PhysRevD.73.124038}, for the case of $\bR^3$.
\end{rem}

A simple closed curve in $\bR^4$ will be referred to as a loop and a set of simple closed curves in $\bR^4$, $L$ will be referred to as a hyperlink. A hyperlink is oriented if we assigned an orientation to its components. We say $L$ is a time-like hyperlink, if given any 2 distinct points $p\equiv (x_0, x_1, x_2, x_3), q\equiv (y_0, y_1, y_2, y_3) \in L$, $p \neq q$, we have
\begin{itemize}
  \item $\sum_{i=1}^3(x_i - y_i)^2 > 0$;
  \item if there exists $i, j$, $i \neq j$ such that $x_i = y_i$ and $x_j = y_j$, then $x_0 - y_0 \neq 0$.
\end{itemize}
Throughout this article, all our hyperlinks in consideration will be time-like. By definition,
it is necessary that $\pi_a(L)$, $a=0, 1, 2, 3$, are all links inside their respective 3-dimensional subspace $\pi_a(\bR^4)\subset \bR^4$ for a time-like hyperlink $L$. See \cite{EH-Lim06}.

\begin{rem}
We wish to point out that for a time-like hyperlink, we are not saying that any distinct points $p, q \in \bR^4$ in the hyperlink must be (Minkowski) time-like separated. In fact, in quantum gravity, there is no preferred choice of metric, so saying that two points are time-like separated makes absolutely no sense at all.
\end{rem}

Consider 2 different hyperlinks $\uL\equiv \{\ul^v\}_{v=1}^{\un}$ and $\oL \equiv \{\ol^u\}_{u=1}^{\on}$ in $\bR^4$, which will be termed as geometric and matter hyperlink respectively. The symbols $u, \bar{u}, v, \bar{v}$ will be indices, taking values in $\mathbb{N}$. They will keep track of the loops in our hyperlinks $\uL$ and $\oL$. The symbols $\un$ and $\on$ will always refer to the number of components in $\uL$ and $\oL$ respectively.

Given two hyperlinks $\oL$ and $\uL$, we also assume that together (by using ambient isotopy if necessary), they form another hyperlink with $\on + \un$ components. Denote this new hyperlink by $\chi(\oL, \uL) \equiv \chi(\{\ol^u\}_{u=1}^{\on}, \{\ul^v\}_{v=1}^{\un})$, assumed to be time-like.

Consider an oriented hyperlink $\chi(\ol, \ul)$, made up of 2 distinct oriented loops $\ol$ and $\ul$. In \cite{EH-Lim03} or \cite{EH-Lim06}, we defined the hyperlinking number between $\ol$ and $\ul$ in $\bR \times\bR^3$, denoted as ${\rm sk}(\ol, \ul)$, to distinguish from the linking number between 2 simple closed curves in $\bR^3$. More generally, define for each $u=1, \ldots, \on$, \beq {\rm sk}(\ol^u, \uL):= \sum_{v=1}^{\un}{\rm sk}(\ol^u, \ul^v), \nonumber\eeq calculated from $\chi(\oL, \uL)$.

\begin{rem}
The hyperlinking number is well-defined for time-like loops. See \cite{EH-Lim03}.
\end{rem}

Using the above basis, define \beq \mathcal{E}^\pm := \sum_{i=1}^3\breve{e}_i \in \mathfrak{su}(2). \label{e.sux.1} \eeq

Let $\rho^\pm: \mathfrak{su}(2) \rightarrow {\rm End}(V^\pm)$ be an irreducible finite dimensional representation, indexed by half-integer and integer values $j_{\rho^\pm} \geq 0$. The representation $\rho: \mathfrak{su}(2) \times \mathfrak{su}(2) \rightarrow {\rm End}(V^+) \times {\rm End}(V^-)$ will be given by $\rho = (\rho^+, \rho^-)$, with \beq \rho: \alpha_i\hat{E}^{0i} + \beta_j \hat{E}^{\tau(j)} \mapsto \left(\sum_{i=1}^3\alpha_i \rho^+(\breve{e}_i) , \sum_{j=1}^3\beta_j \rho^-(\breve{e}_j) \right). \nonumber \eeq By abuse of notation, we will now write $\rho^+ \equiv (\rho^+, 0)$ and $\rho^- \equiv (0, \rho^-)$ in future and thus $\rho^+(\hat{E}^{0i}) \equiv \rho^+(\breve{e}_i)$, $\rho^-(\hat{E}^{\tau(j)}) \equiv \rho^-(\breve{e}_j)$. Without loss of generality, we also assume that $\rho^\pm(\hat{E})$ is skew-Hermitian for any $\hat{E} \in \mathfrak{su}(2)$.

Color the matter hyperlink, i.e. choose a representation $\rho_u: \mathfrak{su}(2) \times \mathfrak{su}(2) \rightarrow {\rm End}(V_u^+) \times {\rm End}(V_u^-)$ for each component $\ol^u$, $u=1, \ldots, \on$, in the hyperlink $\oL$. Note that we do not color $\uL$, i.e. we do not choose a representation for $\uL$.

Let $q \in \bR$ be known as a charge. Define
\begin{align*}
V(\{\ul^v\}_{v=1}^{\underline{n}})(e) :=& \exp\left[ \sum_{v=1}^{\un} \int_{\ul^v} \sum_{\gamma=0}^3 B^i_\gamma \otimes dx_i\right], \\
W(q; \{\ol^u, \rho_u\}_{u=1}^{\on})(\omega) :=& \prod_{u=1}^{\on}\Tr_{\rho_u}  \mathcal{T} \exp\left[ q\int_{\ol^u} A^i_{\alpha\beta} \otimes dx_i\otimes \hat{E}^{\alpha\beta}  \right].
\end{align*}
Here, $\mathcal{T}$ is the time-ordering operator as defined in \cite{CS-Lim02} and $\Tr_{\rho}$ means taking the trace with respect to the representation $\rho$. We sum over repeated indices, with $i$ taking values in 1, 2 and 3; $\alpha\neq \beta$ take values in 0, 1, 2, 3 and $B^i_\gamma, A^i_{\alpha\beta} \in \overline{\mathcal{S}}_\kappa(\bR^4)$.

\begin{rem}
\begin{enumerate}
  \item The term $\mathcal{T} \exp\left[ q\int_{\ol^u} \omega \right]$ is known as a holonomy operator along a loop $\ol^u$, for a spin connection $\omega$.
  \item Witten in \cite{MR990772} used the above holonomy operator in $\bR^3$ to obtain knot invariants in $\bR^3$, by quantizing the Chern-Simons theory, using the path integral approach.
  \item The authors in \cite{PhysRevLett.61.1155} also defined a path integral, involving the holonomy operator, to quantize gravity and they argued that the result should be related to knot invariants.
\end{enumerate}

\end{rem}

Fix a closed and bounded orientable surface, with or without boundary, denoted by $S$, inside $\bR \times\bR^3$. We can write $S$ as a finite disjoint union of surfaces, and without any loss of generality, we assume we can project each such smaller surface to be a surface inside $\bR^3$ or in $\bR \times \Sigma_i$, $i=1,2,3$. Furthermore, we insist the geometric hyperlink $\uL$ is disjoint with $S$ and $\pi_a(\uL)$ intersect at most finitely many points inside $\pi_a(S)$, $a=0, 1, 2, 3$.

Define the following functional, taking values in $\mathfrak{su}(2) \times \mathfrak{su}(2)$,
\begin{align}
F_S(\omega) :=& \frac{1}{2}\int_S \frac{\partial A_{\alpha\beta}^i}{\partial x_0} \otimes dx_0 \wedge dx_i \otimes \hat{E}^{\alpha\beta} + \frac{\partial A_{\alpha \beta}^j}{\partial x_i} \otimes dx_i \wedge dx_j \otimes \hat{E}^{\alpha \beta} \nonumber \\
&\ \ \ \ \ \ + A_{\alpha\beta}^iA_{\gamma\mu}^j\otimes dx_i \wedge dx_j \otimes [\hat{E}^{\alpha\beta}, \hat{E}^{\gamma\mu}]. \label{e.r.2}
\end{align}
Note that we do not specify any representation for $\hat{E}^{\alpha\beta}$ and the functional $F_S(\omega)$, referred to as the total curvature of $S$, is defined by integrating the curvature 2-form $R$, over a surface $S$.

We are going to quantize $F_S$ into an operator, via the path integral, given by \beq \frac{1}{Z}\int_{\omega \in L_\omega,\ e \in L_e}F_S(\omega) V(\{\ul^v\}_{v=1}^{\un})(e) W(q; \{\ol^u, \rho_u\}_{u=1}^{\on})(\omega)  e^{i S_{EH}(e, \omega)}\ De D\omega, \label{e.ehc.1} \eeq whereby $De$ and $D\omega$ are Lebesgue measures on $L_e$ and $L_\omega$ respectively and \beq Z = \int_{\omega \in L_\omega,\ e \in L_e}e^{i S_{EH}(e, \omega)}\ De D\omega. \nonumber \eeq Notice that we are going to average $F_S$ over all possible $\omega \in L_\omega$, weighted against the exponential of the Einstein-Hilbert action.

\begin{rem}
\begin{enumerate}
  \item A similar path integral expression is given in \cite{PhysRevLett.61.1155}.
  \item We do not specify the representation in $F_S$ and so the path integral takes values in $\mathfrak{su}(2) \times \mathfrak{su}(2)$.
  \item In quantum gravity, curvature is treated as a dynamical variable. A quantized curvature operator at each point $p \in \bR^4$ is not defined, as commented in \cite{PhysRevD.73.124038}.
  \item If we regard the curvature as field, then it was explained in \cite{streater} that the curvature $R$ do not yield a well-defined operator; rather only smeared fields will yield quantized operators. In this case, we smear the curvature over a surface $S$, given by $F_S$.
\end{enumerate}

\end{rem}

In \cite{EH-Lim02}, we derived and defined the path integral given by Expression \ref{e.ehc.1} as \beq \lim_{\kappa \rightarrow \infty}\left[(\bar{A}_\kappa^+ + \bar{A}_\kappa^-) + \bar{B}_\kappa + (\bar{C}_\kappa^+ + \bar{C}_\kappa^-)\right]\prod_{u=1}^{\on}\left[ \Tr_{\rho_u^+}\ \hat{\mathcal{W}}_\kappa^+(q; \ol^u, \uL) + \Tr_{\rho_u^-}\ \hat{\mathcal{W}}_\kappa^-(q; \ol^u, \uL) \right], \label{e.sux.4} \eeq the terms $\bar{A}_\kappa^\pm$, $\bar{B}_\kappa$ and $\bar{C}_\kappa^\pm$ will all be defined by Equations (\ref{e.c.4a}), (\ref{e.c.4c}) and (\ref{e.c.4d}) respectively.

Refer to Notations \ref{n.n.3} and \ref{n.k.1}. Note that $\hat{\mathcal{W}}_\kappa^\pm(q; \ol^u, \uL)$ was defined in \cite{EH-Lim02} as ($d\hat{s} \equiv dsd\bar{s}$)
\begin{align}
\hat{\mathcal{W}}_\kappa^\pm(q; \ol^u, \uL) := \exp\left[ \mp\frac{iq}{4}\frac{\kappa^3}{4\pi}\sum_{v=1}^{\underline{n}}\int_{I^2}\ d\hat{s}\ \epsilon^{ijk}\left\langle  p_\kappa^{\vec{y}^u_s}, p_\kappa^{\vec{\varrho}^v_{\bar{s}}}\right\rangle_k  y^{u,\prime}_{i,s}\varrho^{v,\prime}_{j,\bar{s}} \otimes \mathcal{E}^\pm\right], \label{e.h.5}
\end{align}
whereby $\mathcal{E}^\pm$ was defined in Equation (\ref{e.sux.1}). And $\epsilon^{ijk} \equiv \epsilon_{ijk}$ be defined on the set $\{1,2,3\}$, by \beq \epsilon^{123} = \epsilon^{231} = \epsilon^{312} = 1,\ \ \epsilon^{213} = \epsilon^{321} = \epsilon^{132} = -1, \nonumber \eeq if $i,j,k$ are all distinct; 0 otherwise.

\begin{rem}
\begin{enumerate}
  \item When $S$ is the empty set, we define $F_\emptyset \equiv 1$, so we write Expression \ref{e.ehc.1} as $Z(q; \chi(\oL, \uL))$, termed as the Wilson Loop observable of the colored hyperlink $\chi(\oL, \uL)$.
  \item For each bounded surface $S$, we can define an operator, henceforth be denoted by $\hat{F}_S$ and we will write Expression \ref{e.ehc.1} as $\hat{F}_S[Z(q; \chi(\oL, \uL))]$. One should view the operator $\hat{F}_S$ as acting on the Wilson Loop observable $Z(q; \chi(\oL, \uL))$.
\end{enumerate}
\end{rem}

Using the lemma in the appendix found in \cite{EH-Lim02}, one can show that \beq  \lim_{\kappa \rightarrow \infty}\Tr_{\rho_u^\pm}\ \hat{\mathcal{W}}_\kappa^\pm(q; \ol^u, \uL) = \Tr_{\rho^\pm_u}\ \exp[\mp\pi iq\ {\rm sk}(\ol^u, \uL) \cdot \mathcal{E}^\pm]. \label{e.w.5}  \eeq The reader can also refer to \cite{EH-Lim03} for a detailed proof.

Hence
\begin{align}
Z(q; \chi(\oL, \uL) )
:=& \lim_{\kappa \rightarrow \infty}\prod_{u=1}^{\on}\left[\Tr_{\rho_u^+}\ \hat{\mathcal{W}}_\kappa^+(q; \ol^u, \uL) + \Tr_{\rho_u^-}\ \hat{\mathcal{W}}_\kappa^-(q; \ol^u, \uL) \right]\nonumber\\
=& \prod_{u=1}^{\on}\left( \Tr_{\rho^+_u}\ \exp[-\pi iq\ {\rm sk}(\ol^u, \uL) \cdot \mathcal{E}^+] +
\Tr_{\rho^-_u}\ \exp[\pi iq\ {\rm sk}(\ol^u, \uL) \cdot \mathcal{E}^-] \right).\label{e.sux.3}
\end{align}
Note that $\Tr_{\rho^\pm_u}$ means take the trace. See \cite{EH-Lim02} or \cite{EH-Lim03}.
In this article, we will compute the limits of the remaining 5 terms in Expression \ref{e.sux.4}, as $\kappa$ goes to infinity. This will give us the main result Theorem \ref{t.c.1} in this article.



From our main Theorem \ref{t.c.1}, we see that the path integral is computed using the linking number between an oriented geometric hyperlink $\uL$ and an oriented surface $S$. We refer to the definition of the linking number between a hyperlink with a surface given in either \cite{EH-Lim03} or \cite{EH-Lim06}. It is a topological invariant, hence it is invariant under any diffeomorphism of $\bR \times \bR^3$, not just Lorentz transformation. See \cite{greensite2011introduction}. Note that in any diffeomorphism of $\bR^4$, we insist that the matter and geometric hyperlinks must remain time-like under this transformation.

\begin{rem}
In this article, we fixed our matter and geometric hyperlinks and surface $S$, thus we do not consider ambient isotopic classes of loops and surface $S$.
\end{rem}

As a consequence, curvature is quantized via the Expression \ref{e.ehc.1}. As an application, we will explain how quantum gravity solves certain inconsistencies in General Relativity as outlined in \cite{Thiemann:2002nj}.

Compare this with our separate article \cite{EH-Lim03}, in which we showed that the area operator is computed using the piercing number between a hyperlink and a surface. The physical interpretation will be that the area operator sums up the impact of the particles colliding with a surface $S$, taking into account the particles' momentum from boost and rotation.

\subsection{Limit of $\bar{A}_\kappa^\pm$}\label{ss.f.1a}

\begin{notation}(Parametrization of curves)\label{n.n.3}\\
Let $\vec{y}^u \equiv (y_0^u, y_1^u, y_2^u, y_3^u) : I := [0,1] \rightarrow \bR \times \bR^3$ be a parametrization of a loop $\ol^u \subset \oL$, $u=1, \ldots, \on$. We will write $y^u(s) = (y_1^u(s), y_2^u(s), y_3^u(s))$ and $\vec{y}^u(s) \equiv \vec{y}_s^u$. We will also write $\vec{y}^u = (y^u_0, y^u)$. Similarly, choose a parametrization $\vec{\varrho}^{v}: I \rightarrow \bR \times \bR^3$, $v = 1, \ldots, \un$, for each loop $\ul^v \subset \uL$. When the loop is oriented, we will often choose a parametrization which is consistent with the assigned orientation.
\end{notation}

\begin{notation}\label{n.s.3}
Choose an orientable, closed and bounded surface $S \subset \bR^4$, with or without boundary. If it has a boundary $\partial S$, then $\partial S$ is assumed to be a time-like hyperlink. Do note that we allow $S$ to be disconnected, with finite number of components. Parametrize it using \beq \vec{\sigma}: \hat{t} \equiv (t, \bar{t}) \in I^2 \mapsto \left(\sigma_0(t,\bar{t}), \sigma_1(t,\bar{t}), \sigma_2(t,\bar{t}), \sigma_3(t,\bar{t}) \right)\in  \bR^4 \nonumber \eeq and let \beq J_{\alpha\beta} = \frac{\partial \sigma_\alpha}{\partial t} \frac{\partial \sigma_\beta}{\partial \bar{t}} - \frac{\partial \sigma_\alpha}{\partial \bar{t}} \frac{\partial \sigma_\beta}{\partial t}. \nonumber \eeq

When we project $S$ inside $\bR^3$ as $\pi_0(S)$, we can parametrize it using $\sigma \equiv (\sigma_1, \sigma_2, \sigma_3)$. Let
\begin{align*}
K_\sigma(\hat{t}) \equiv& K_\sigma(t,\bar{t}) := (J_{01}(t,\bar{t}), J_{02}(t,\bar{t}), J_{03}(t,\bar{t})),\\
J_\sigma(\hat{t}) \equiv& J_\sigma(t,\bar{t}) := \frac{\partial}{\partial t}\sigma(t, \bar{t}) \times  \frac{\partial}{\partial\bar{t}}\sigma(t, \bar{t}) \equiv (J_{23}(\hat{t}), J_{31}(\hat{t}), J_{12}(\hat{t})).
\end{align*}
If $S$ is assigned an orientation, we will choose a parametrization such that the inherited orientation on $\pi_0(S) \subset \bR^3$ is consistent with the vector $J_\sigma$.
\end{notation}

Let $\vec{y} \equiv (y_0, y) \in \bR^4$, whereby $y \equiv (y_1, y_2, y_3) \in \bR^3$. We will write
\beq \hat{y}_i =
 \left\{
  \begin{array}{ll}
    (y_2, y_3), & \hbox{$i=1$;} \\
    (y_1, y_3), & \hbox{$i=2$;} \\
    (y_1, y_2), & \hbox{$i=3$.}
  \end{array}
\right. \nonumber \eeq

If $x \in \bR^n$, we will write $(p_\kappa^x)^2$ to denote the $n$-dimensional Gaussian function, center at $x$, variance $1/\kappa^2$. For example, \beq p_\kappa^x(\cdot) = \frac{\kappa^2}{2\pi}e^{-\kappa^2|\cdot - x|^2/4},\ x \in \bR^4. \nonumber \eeq We will also write $(q_\kappa^x)^2$ to denote the 1-dimensional Gaussian function, i.e. \beq q_\kappa^x(\cdot) = \frac{\sqrt{\kappa}}{(2\pi)^{1/4}}e^{-\kappa^2 (\cdot - x)^2/4}. \nonumber \eeq

For $x, y \in \bR^2$, we write \beq \left\langle p_\kappa^x, p_\kappa^y \right\rangle = \int_{z\in \bR^2} \frac{\kappa}{\sqrt{2\pi}}e^{-\kappa^2|z - x|^2/4} \frac{\kappa}{\sqrt{2\pi}}e^{-\kappa^2|z - y|^2/4} dz, \nonumber \eeq i.e. we integrate over Lebesgue measure on $\bR^2$. More generally, given $f, g \in C(\bR^n)$, we will write \beq \langle f, g \rangle \equiv \int_{\bR^n} f\cdot g\ d\lambda, \nonumber \eeq whereby $\lambda$ is Lebesgue measure.

For $x = (x_0,x_1, x_2, x_3)$, write \beq x(s_a) :=
\left\{
  \begin{array}{ll}
    (s_0,x_1, x_2, x_3), & \hbox{$a=0$;} \\
    (x_0,s_1, x_2, x_3), & \hbox{$a=1$;} \\
    (x_0,x_1, s_2, x_3), & \hbox{$a=2$;} \\
    (x_0,x_1, x_2, s_3), & \hbox{$a=3$.}
  \end{array}
\right. \nonumber \eeq

Let $\partial_a \equiv \partial/\partial x_a$ be a differential operator. There is an operator $\partial_a^{-1}$ acting on a dense subset in $\overline{\mathcal{S}}_\kappa(\bR^4)$, \beq (\partial_a^{-1}f)(x) := \frac{1}{2}\int_{-\infty}^{x_a} f(x(s_a))\ ds_a - \frac{1}{2}\int_{x_a}^{\infty} f(x(s_a))\ ds_a,\ f \in \overline{\mathcal{S}}_\kappa(\bR^4). \label{e.d.1} \eeq Here, $x_a \in \bR$. Notice that $\partial_a\partial_a^{-1}f \equiv f$ and $\partial_a^{-1}f$ is well-defined provided $f$ is in $L^1$.

\begin{notation}\label{n.k.1}
For each $i = 1, 2, 3$, write
\beq \left\langle p_\kappa^{\vec{x}}, p_\kappa^{\vec{y}} \right\rangle_i =
\left\langle p_\kappa^{\hat{x}_{i}}, p_\kappa^{\hat{y}_{i}} \right\rangle \left\langle q_\kappa^{x_{i}}, \kappa\partial_0^{-1}q_\kappa^{y_{i}} \right\rangle\left\langle \partial_0^{-1}q_\kappa^{x_{0}}, q_\kappa^{y_{0}}\right\rangle. \nonumber \eeq

Here, \beq \partial_0^{-1}q_\kappa^{x_{0}}(t) \equiv \frac{1}{2}\int_{-\infty}^t q_\kappa^{x_{0}}(\tau)\ d\tau -
\frac{1}{2}\int_{t}^\infty q_\kappa^{x_{0}}(\tau)\ d\tau. \nonumber \eeq

Note that $\left\langle \partial_0^{-1}q_\kappa^{x_{0}}, q_\kappa^{y_{0}}\right\rangle \equiv \left\langle q_\kappa^{y_{0}}, \partial_0^{-1}q_\kappa^{x_{0}} \right\rangle$  means we integrate $\partial_0^{-1}q_\kappa^{x_{0}} \cdot q_\kappa^{y_{0}}$ over $\bR$, using Lebesgue measure. It is well-defined because $q_\kappa^{x_0}$ is in $L^1$.
\end{notation}

In \cite{EH-Lim02}, we defined
\begin{align}
\bar{A}_\kappa^\pm := \mp\frac{i\kappa^3}{32\pi\sqrt{4\pi}}\sum_{v=1}^{\un}\int_{I^3}\ \Bigg\{ &\left\langle  p_\kappa^{\vec{\sigma}(\hat{t})}, \kappa\left[\varrho^{v,\prime}_{3,s}\partial_2^{-1} - \varrho^{v,\prime}_{2,s}\partial_3^{-1}\right]p_\kappa^{\vec{\varrho}^{v}_s}\right\rangle J_{01}(\hat{t})\nonumber \\
+& \left\langle p_\kappa^{\vec{\sigma}(\hat{t})}, \kappa\left[\varrho^{v,\prime}_{1,s}\partial_3^{-1} - \varrho^{v,\prime}_{3,s}\partial_1^{-1}\right]p_\kappa^{\vec{\varrho}^{v}_s}\right\rangle J_{02}(\hat{t})\nonumber\\
+&\left\langle p_\kappa^{\vec{\sigma}(\hat{t})}, \kappa\left[\varrho^{v,\prime}_{2,s}\partial_1^{-1} - \varrho^{v,\prime}_{1,s}\partial_2^{-1}\right]p_\kappa^{\vec{\varrho}^{v}_s}\right\rangle J_{03}(\hat{t})
\Bigg\}\ ds d\hat{t} \otimes \mathcal{F}^{\pm},\label{e.c.4a}
\end{align}
whereby $\mathcal{F}^+ = \sum_{i=1}^3 \hat{E}^{0i}$, $\mathcal{F}^- = \sum_{i=1}^3 \hat{E}^{\tau(i)}$.

To compute its limit as $\kappa$ goes to infinity, we need the following lemma.
We will denote the linking number between an oriented surface $S$ and an oriented loop $l$ by ${\rm lk}(l, S)$.

\begin{lem}
Let $S$ be a surface in $\bR \times \bR^3$ and refer to Notation \ref{n.s.3}.\\ Suppose $(i,j,k) \in \{(1,2,3), (3,1,2), (2,3,1)\}$. We have
\begin{align*}
\frac{\kappa^3}{32\pi}&\int_{I^3}\left\langle  p_\kappa^{\vec{\sigma}(\hat{t})}, \kappa\partial_j^{-1}p_\kappa^{\vec{\varrho}^{v}_s}\right\rangle \left[\varrho^{v,\prime}_{k,s}J_{0i}(\hat{t}) -
\varrho^{v,\prime}_{i,s}J_{0k}(\hat{t})\right]
ds d\hat{t} \longrightarrow \pi \cdot {\rm lk}(l^v, S).
\end{align*}
\end{lem}

\begin{proof}
Fix a $j=1, 2, 3$, which determines $i$ and $k$. Note that we can split the surface $S$ into a finite disjoint union of surfaces $\bigcup_{u}S_u$, such that at most one curve in $\pi_j(\ul^v)$ pierce each such surface $S_u$, and such an intersection point is referred to as a piercing. And because we will in our sequel to this article, consider equivalence class of curves, we may and will assume that at the piercing $p \in \bR \times \Sigma_j$, the intersection of $\pi_j(\ul^v)$ with $S_u$, $\varrho^{v,\prime}_0 = 0$ in a small neighborhood around $p$. That is, a small segment of $\pi_j(\ul^v)$ that contains $p$ lies in a plane that is parallel to $\Sigma_j$.

Thus, without loss of generality, we assume that $\omega: I^3 \rightarrow \bR^3$, \beq \omega(t, \bar{t}, s) :=
(\sigma_k(\hat{t}) -\varrho_k^v(s),\ \sigma_0(\hat{t}) -\varrho_0^v(s),\ \sigma_i(\hat{t}) -\varrho_i^v(s)),\ \hat{t} \equiv (t, \bar{t}), \nonumber \eeq is diffeomorphic to an open ball in $\bR^3$. Notice that here, we use the time-axis and the $\Sigma_j$ plane to form $\bR^3$. And if $p = (\varrho_k^v(s_0), \varrho_0^v(s_0), \varrho_i^v(s_0))$, then there exists an $\epsilon > 0$ such that for all $s \in (s_0-\epsilon, s_0 + \epsilon)$, $\varrho_0^{v,\prime}(s) = 0$. So without loss of generality, we assume that on $I$, $\varrho_0^{v,\prime}(s) = 0$. However, on $I^2$, $\partial \sigma_0/\partial t, \partial \sigma_0/\partial \bar{t} \neq 0$.

Refer to Notations \ref{n.n.3} and \ref{n.k.1}. From the lemma in the appendix in \cite{EH-Lim02} or Lemma 5 in \cite{CS-Lim03}, it is straightforward to show that
\begin{align*}
\frac{1}{\sqrt{2\pi}}\left\langle  p_\kappa^{\vec{\sigma}(\hat{t})}, \kappa\partial_j^{-1}p_\kappa^{\vec{\varrho}^{v}_s}\right\rangle &= \frac{1}{\sqrt{2\pi}}\left\langle q_\kappa^{\sigma_{j}(\hat{t})}, \kappa\partial_0^{-1}q_\kappa^{\varrho_{j,s}^{v}} \right\rangle \cdot \left\langle p_\kappa^{\hat{\sigma}_{j}(\hat{t})}, p_\kappa^{\hat{\varrho}_{j,s}^v} \right\rangle \cdot \left\langle q_\kappa^{\sigma_{0}(\hat{t})}, q_\kappa^{\varrho_{0,s}^v} \right\rangle \\
&= \left\langle q_\kappa^{\sigma_{j}(\hat{t})}, \frac{\kappa}{\sqrt{2\pi}}\partial_0^{-1}q_\kappa^{\varrho_{j,s}^{v}} \right\rangle e^{-\kappa^2|(\sigma_0(\hat{t}) - \varrho_0^v(s),  \sigma_i(\hat{t}) -\varrho_i^v(s), \sigma_k(\hat{t})-\varrho_k^v(s))|^2/8}.
\end{align*}

Let $R = \omega(I^3)$ and let $z: \omega \mapsto \sigma_j(\hat{t})$, $x : \omega \mapsto \varrho_j^v(s)$. Suppose that \beq (\varrho^{v,\prime}_k, \varrho^{v,\prime}_0, \varrho^{v,\prime}_i ) \cdot (J_{0i}, J_{ik}, J_{k0}) \equiv  (\varrho^{v,\prime}_kJ_{0i} + \varrho^{v,\prime}_0J_{ik} + \varrho^{v,\prime}_iJ_{k0})> 0. \nonumber \eeq If is negative, then our final result below will pick up an extra negative sign.

Using a change in variable argument, we see that ($J_{0i} = -J_{i0}$.)
\begin{align}
\frac{\kappa^3}{32\pi}&\int_{I^3}\left\langle  p_\kappa^{\vec{\sigma}(\hat{t})}, \kappa\partial_j^{-1}p_\kappa^{\vec{\varrho}^{v}_s}\right\rangle \left[\varrho^{v,\prime}_{k,s}J_{0i}(\hat{t})
- \varrho^{v,\prime}_{i,s}J_{0k}(\hat{t})\right]
ds d\hat{t} \nonumber\\
=& \pi\frac{\kappa^3}{16\pi\sqrt{2\pi}}\int_{R}\left\langle q_\kappa^{z(\omega)}, \frac{\kappa}{\sqrt{2\pi}}\partial_0^{-1}q_\kappa^{x(\omega)} \right\rangle \cdot e^{-\kappa^2|\omega|^2/8}
d\omega. \label{e.c.1}
\end{align}

From Item 1 in the lemma in the appendix in \cite{EH-Lim02}, or from Item 2 in Lemma 5 in \cite{CS-Lim03}, we see that
\beq \left\langle q_\kappa^{z(\omega)}, \frac{\kappa}{\sqrt{2\pi}}\partial_0^{-1}q_\kappa^{x(\omega)} \right\rangle = \sgn(z(\omega) - x(\omega)). \nonumber \eeq And \beq \frac{\kappa^3}{16\pi\sqrt{2\pi}}\int_R e^{-\kappa^2|\omega|^2/8} d\omega \rightarrow 1, \nonumber \eeq if $0 \in R$, 0 otherwise. When $0 \in R$, it means that the curve $(\varrho_k, \varrho_0, \varrho_i)$ pierce the surface $\pi_j(S)$.

Putting all together, we see that the RHS of Equation (\ref{e.c.1}) converges to $\pi \cdot \sgn(z(\omega) - x(\omega))$, which is equal to $\pi \varepsilon(p)$, $\varepsilon(p)$ is the algebraic piercing number of $p$ defined in \cite{EH-Lim03} or \cite{EH-Lim06}.
\end{proof}

Suppose $L = \{l^1, \ldots, l^n\}$ is an oriented hyperlink in $\bR^4$. Then for an oriented surface $S$, we define the linking number between $L$ and $S$ as \beq {\rm lk}(L, S) := \sum_{u=1}^n {\rm lk}(l^u, S). \nonumber \eeq

\begin{cor}\label{c.c.3}
The sum $\bar{A}_\kappa^+ + \bar{A}_\kappa^-$ from Expression \ref{e.c.4a} converges to \beq -\frac{3i\sqrt\pi}{2} {\rm lk}(\uL, S) \otimes \left[ \sum_{i=1}^3 \hat{E}^{0i} - \sum_{j=1}^3\hat{E}^{\tau(j)} \right]. \nonumber \eeq
\end{cor}

\subsection{Limit of $\bar{B}_\kappa$}\label{ss.f.1b}

Refer to Notations \ref{n.n.3}, \ref{n.s.3} and \ref{n.k.1}.

In \cite{EH-Lim02}, we defined
\begin{align}
\bar{B}_\kappa := \frac{i\kappa^3}{32\pi\sqrt{4\pi}}&\int_{I^3}\  \kappa\left\langle \partial_0^{-1} p_\kappa^{\vec{\sigma}(\hat{t})}, p_\kappa^{\vec{\varrho}^{v}_s}\right\rangle \varrho_s^{v,\prime}\cdot J_\sigma(\hat{t}) \otimes \sum_{i=1}^3\hat{E}^{0i}\ ds d\hat{t}, \nonumber \\
-\frac{i\kappa^3}{32\pi\sqrt{4\pi}}&\int_{I^3}\ \kappa\left\langle \partial_0^{-1} p_\kappa^{\vec{\sigma}(\hat{t})}, p_\kappa^{\vec{\varrho}^{v}_s}\right\rangle \varrho_s^{v,\prime}\cdot J_\sigma(\hat{t}) \otimes \sum_{j=1}^3\hat{E}^{\tau(j)}\ ds d\hat{t}, \label{e.c.4c}
\end{align}
In this subsection, we will compute its limit as $\kappa$ goes to infinity.

\begin{lem}
Let $S$ be a surface in $\bR \times \bR^3$. Refer to Notation \ref{n.s.3}. We have
\begin{align*}
\frac{\kappa^3}{32\pi}\int_{I^3}\  &\kappa\left\langle \partial_0^{-1}  p_\kappa^{\vec{\sigma}(\hat{t})}, p_\kappa^{\vec{\varrho}^{v}_s}\right\rangle \varrho_s^{v,\prime}\cdot J_\sigma(\hat{t})\ ds d\hat{t} \longrightarrow -\pi \cdot {\rm lk}(\ul^v, S).
\end{align*}
\end{lem}

\begin{proof}
Note that we can split the surface $\pi_0(S)$ into a finite disjoint union of surfaces, such that $\ul^v$ pierce each such surface at most once.

Thus, without loss of generality, we assume that $\omega: I^3 \rightarrow \bR^3$, \beq \omega(t, \bar{t}, s) := \sigma(\hat{t})-\varrho^v(s) ,\ \hat{t} \equiv (t, \bar{t}), \nonumber \eeq is diffeomorphic to an open ball in spatial $\bR^3$.

It is straightforward to show using Item 3 in the lemma in the appendix in \cite{EH-Lim02}, or using Item 2 in Lemma 5 in \cite{CS-Lim03}, that
\beq \kappa\left\langle \partial_0^{-1} p_\kappa^{\vec{\sigma}(\hat{t})},  p_\kappa^{\vec{\varrho}^{v}_s}\right\rangle
= \left\langle \kappa\partial_0^{-1}q_\kappa^{\sigma_{0}(\hat{t})},  q_\kappa^{\varrho_0^v(s)} \right\rangle
e^{-\kappa^2|\sigma(\hat{t}) - \varrho^v(s)|^2/8}.
\nonumber \eeq

Let $R = \omega(I^3)$ and let $z: \omega \mapsto \sigma_0(\hat{t})$, $x : \omega \mapsto \varrho_0^v(s)$. Using a change in variable argument, we see that
\begin{align}
\frac{\kappa^3}{32\pi}&\int_{I^3}\left\langle \kappa\partial_0^{-1} p_\kappa^{\vec{\sigma}(\hat{t})}, p_\kappa^{\vec{\varrho}^{v}_s}\right\rangle \varrho_s^{v,\prime}\cdot J_\sigma(\hat{t})
ds d\hat{t} \nonumber\\
=& {\rm sgn}(\varrho^{v,\prime}\cdot J_\sigma)\frac{\pi\kappa^3}{16\pi\sqrt{2\pi}}\int_{R}\left\langle \frac{\kappa}{\sqrt{2\pi}}\partial_0^{-1}q_\kappa^{z(\omega)}, q_\kappa^{x(\omega)} \right\rangle \cdot e^{-\kappa^2|\omega|^2/8}
d\omega. \label{e.c.2}
\end{align}

From Item 1 in the lemma in the appendix in \cite{EH-Lim02}, or using Item 2 in Lemma 5 in \cite{CS-Lim03}, we see that
\beq \left\langle \frac{\kappa}{\sqrt{2\pi}}\partial_0^{-1}q_\kappa^{z(\omega)}, q_\kappa^{x(\omega)} \right\rangle = \sgn(x(\omega) - z(\omega)) = -\sgn(z(\omega) - x(\omega)). \nonumber \eeq And \beq \frac{\kappa^3}{16\pi\sqrt{2\pi}}\int_R e^{-\kappa^2|\omega|^2/8} d\omega \rightarrow 1, \nonumber \eeq if $0 \in R$, 0 otherwise. When $0 \in R$, it means that the curve in $\bR^3$, parametrized by $\varrho^v$, pierce the surface $\pi_0(S)$.

Putting all together, we see that the RHS of Equation (\ref{e.c.2}) converges to $\mp \pi$, depending on the signs of ${\rm sgn}(\varrho^{v,\prime}\cdot J_\sigma)$ and $\sgn(z(0) - x(0))$, which are equal to the orientation and height of $p\in \bR^3$ respectively as defined in \cite{EH-Lim03} or \cite{EH-Lim06}, where $\pi_0(\ul^v)$ intersects $\pi_0(S)$ at $p$.
\end{proof}

\begin{cor}\label{c.c.2}
We have Expression \ref{e.c.4c} converges to \beq -\frac{i\sqrt{\pi}}{2} {\rm lk}(\uL, S)  \otimes \left[ \sum_{i=1}^3 \hat{E}^{0i} - \sum_{j=1}^3\hat{E}^{\tau(j)} \right]. \nonumber \eeq
\end{cor}

\subsection{Limit of $\bar{C}_\kappa^\pm$}\label{ss.fs3}

Refer to Notations \ref{n.n.3}, \ref{n.s.3} and \ref{n.k.1}. Let $C_3 = \{(1,2,3), (2,3,1), (3,1,2)\}$.

In \cite{EH-Lim02}, we defined
\begin{align}
\bar{C}_\kappa^\pm := \frac{-\kappa^4}{32\pi^2}\sum_{(i,j,k)\in C_3}\sum_{v,\bar{v}=1}^{\un}\int_{I^4}\ \Bigg\{ &\left\langle  \partial_0^{-1}p_\kappa^{\vec{\sigma}(\hat{t})}, \kappa\left[\varrho^{v,\prime}_{k,s}\partial_j^{-1} - \varrho^{v,\prime}_{j,s}\partial_k^{-1}\right]p_\kappa^{\vec{\varrho}^{v}_s}\right\rangle \nonumber\\
\times& \left\langle  \partial_0^{-1}p_\kappa^{\vec{\sigma}(\hat{t})}, \kappa\left[\varrho^{\bar{v},\prime}_{i,\bar{s}}\partial_k^{-1} - \varrho^{\bar{v},\prime}_{k,\bar{s}}\partial_i^{-1}\right]p_\kappa^{\vec{\varrho}^{\bar{v}}_{\bar{s}}}
\right\rangle
\Bigg\}\ J_{ij}(\hat{t})\ d\hat{s}d\hat{t} \otimes \mathcal{F}^\pm.\label{e.c.4d}
\end{align}
Note that $\hat{s} = (s, \bar{s})$ and $\hat{t} = (t, \bar{t})$ and $\mathcal{F}^+ = \sum_{i=1}^3 \hat{E}^{0i}$, $\mathcal{F}^- = \sum_{i=1}^3 \hat{E}^{\tau(i)}$.
In this subsection, we will show that its limit as $\kappa$ goes to infinity is zero.

\begin{lem}\label{l.c.1}
For $(i,j,k) \in C_3$ and $\hat{s} = (s, \bar{s})$, $\hat{t} = (t, \bar{t})$, we have
\begin{align}
\int_{I^4} \kappa^4 &\left\langle  \partial_0^{-1}p_\kappa^{\vec{\sigma}(\hat{t})}, \kappa \varrho^{v,\prime}_{k,s}\partial_j^{-1}p_\kappa^{\vec{\varrho}^{v}_s}\right\rangle
\left\langle \partial_0^{-1} p_\kappa^{\vec{\sigma}(\hat{t})}, \kappa \varrho^{v,\prime}_{i,\bar{s}}\partial_k^{-1}p_\kappa^{\vec{\varrho}^{\bar{v}}_{\bar{s}}}\right\rangle
J_{ij}(\hat{t})\ d\hat{t} d\hat{s}, \label{e.ca.1a}\\
\int_{I^4}\kappa^4 &\left\langle  \partial_0^{-1}p_\kappa^{\vec{\sigma}(\hat{t})}, \kappa \varrho^{v,\prime}_{k,s}\partial_j^{-1}p_\kappa^{\vec{\varrho}^{v}_s}\right\rangle
\left\langle \partial_0^{-1} p_\kappa^{\vec{\sigma}(\hat{t})}, \kappa \varrho^{\bar{v},\prime}_{k,\bar{s}}\partial_i^{-1}p_\kappa^{\vec{\varrho}^{\bar{v}}_{\bar{s}}}\right\rangle
J_{ij}(\hat{t})\ d\hat{t} d\hat{s}, \label{e.ca.1b}\\
\int_{I^4}\kappa^4 &\left\langle  \partial_0^{-1}p_\kappa^{\vec{\sigma}(\hat{t})}, \kappa \varrho^{v,\prime}_{j,s}\partial_k^{-1}p_\kappa^{\vec{\varrho}^{v}_s}\right\rangle
\left\langle \partial_0^{-1} p_\kappa^{\vec{\sigma}(\hat{t})}, \kappa \varrho^{\bar{v},\prime}_{i,\bar{s}}\partial_k^{-1}p_\kappa^{\vec{\varrho}^{\bar{v}}_{\bar{s}}}\right\rangle
J_{ij}(\hat{t})\ d\hat{t} d\hat{s}, \label{e.ca.1c}\\
\int_{I^4}\kappa^4 &\left\langle  \partial_0^{-1}p_\kappa^{\vec{\sigma}(\hat{t})}, \kappa \varrho^{v,\prime}_{j,s}\partial_k^{-1}p_\kappa^{\vec{\varrho}^{v}_s}\right\rangle
\left\langle \partial_0^{-1} p_\kappa^{\vec{\sigma}(\hat{t})}, \kappa \varrho^{\bar{v},\prime}_{k,\bar{s}}\partial_i^{-1}p_\kappa^{\vec{\varrho}^{\bar{v}}_{\bar{s}}}\right\rangle
J_{ij}(\hat{t})\ d\hat{t} d\hat{s},\label{e.ca.1d}
\end{align}
all converge to 0 as $\kappa$ goes to infinity.
\end{lem}

\begin{proof}
We first deal with Expression \ref{e.ca.1a}. Now,
\begin{align*}
\kappa^2&\left\langle  \partial_0^{-1}p_\kappa^{\vec{\sigma}(\hat{t})}, \kappa \varrho^{v,\prime}_{k,s}\partial_j^{-1}p_\kappa^{\vec{\varrho}^{v}_s}\right\rangle \\
&= -\left\langle q_\kappa^{\sigma_0(\hat{t})}, \kappa\partial_0^{-1} q_\kappa^{\varrho_{0,s}^v} \right\rangle
\left\langle q_\kappa^{\sigma_j(\hat{t})}, \kappa\partial_0^{-1} q_\kappa^{\varrho_{j,s}^v} \right\rangle
\kappa e^{-\kappa^2|\sigma_i(\hat{t}) - \varrho_{i,s}^v|^2/8}e^{-\kappa^2|\sigma_k(\hat{t}) - \varrho_{k,s}^v|^2/8} \varrho^{v,\prime}_{k,s}, \label{e.ca.2}
\end{align*}
and
\begin{align*}
\kappa^2&\left\langle  \partial_0^{-1}p_\kappa^{\vec{\sigma}(\hat{t})}, \kappa \varrho^{\bar{v},\prime}_{i,\bar{s}}\partial_k^{-1}p_\kappa^{\vec{\varrho}^{\bar{v}}_{\bar{s}}}
\right\rangle \\
&= -\left\langle q_\kappa^{\sigma_0(\hat{t})}, \kappa\partial_0^{-1} q_\kappa^{\varrho_{0,s}^v} \right\rangle
\left\langle q_\kappa^{\sigma_k(\hat{t})}, \kappa\partial_0^{-1} q_\kappa^{\varrho_{k,\bar{s}}^{\bar{v}}} \right\rangle
\kappa e^{-\kappa^2|\sigma_i(\hat{t}) - \varrho_{i,\bar{s}}^{\bar{v}}|^2/8}e^{-\kappa^2|\sigma_j(\hat{t}) - \varrho_{j,\bar{s}}^{\bar{v}}|^2/8} \varrho^{\bar{v},\prime}_{i,\bar{s}}.
\end{align*}
Note that we used $\langle \partial_0^{-1}f, g \rangle = -\langle f, \partial_0^{-1}g \rangle$.

Since $\langle q_\kappa^x, \kappa \partial_0^{-1} q_\kappa^y \rangle $ is bounded by a constant from Item 1 in the lemma in the appendix in \cite{EH-Lim02}, or from Item 2 in Lemma 5 in \cite{CS-Lim03}, it suffices to prove that
\begin{align*}
\tilde{C}_1:= \int_{I^4}d\hat{t} d\hat{s}\ J_{ij}(\hat{t})& \varrho^{v,\prime}_{k,s}
\varrho^{\bar{v},\prime}_{i,\bar{s}} \\
\times&\ \kappa^2 e^{-\kappa^2|\sigma_i(\hat{t}) - \varrho_{i,s}^v|^2/8}e^{-\kappa^2|\sigma_k(\hat{t}) - \varrho_{k,s}^v|^2/8}
e^{-\kappa^2|\sigma_i(\hat{t}) - \varrho_{i,\bar{s}}^{\bar{v}}|^2/8}e^{-\kappa^2|\sigma_j(\hat{t}) - \varrho_{j,\bar{s}}^{\bar{v}}|^2/8}
\end{align*}
goes to 0 as $\kappa$ goes to infinity.
Now,
\begin{align*}
|\tilde{C}_1| \leq& \kappa^2\int_{I^3}
 e^{-\kappa^2|\sigma_i(\hat{t}) - \varrho_{i,\bar{s}}^{\bar{v}}|^2/8}e^{-\kappa^2|\sigma_j(\hat{t}) - \varrho_{j,\bar{s}}^{\bar{v}}|^2/8}
\int_I \left| e^{-\kappa^2|\sigma_i(\hat{t}) - \varrho_{i,s}^v|^2/8}e^{-\kappa^2|\sigma_k(\hat{t}) - \varrho_{k,s}^v|^2/8} \varrho^{v,\prime}_{k,s} \right| ds\\
& \times \
\left| \varrho^{\bar{v},\prime}_{i,\bar{s}}J_{ij}(\hat{t})\right| \ d\hat{t} d\bar{s} \\
\leq& \kappa^2\int_{I^3}
 e^{-\kappa^2|\sigma_i(\hat{t}) - \varrho_{i,\bar{s}}^{\bar{v}}|^2/8}e^{-\kappa^2|\sigma_j(\hat{t}) - \varrho_{j,\bar{s}}^{\bar{v}}|^2/8}
\int_I \left| e^{-\kappa^2|\sigma_k(\hat{t}) - \varrho_{k,s}^v|^2/8} \varrho^{v,\prime}_{k,s} \right| ds\\
& \times \
\left| \varrho^{\bar{v},\prime}_{i,\bar{s}}J_{ij}(\hat{t})\right| \ d\hat{t} d\bar{s} \\
\leq& 2\kappa\sqrt{2\pi}\int_{I^3}
 e^{-\kappa^2|\sigma_i(\hat{t}) - \varrho_{i,\bar{s}}^{\bar{v}}|^2/8}e^{-\kappa^2|\sigma_j(\hat{t}) - \varrho_{j,\bar{s}}^{\bar{v}}|^2/8}\ \left| \varrho^{\bar{v},\prime}_{i,\bar{s}}J_{ij}(\hat{t})\right| \ d\hat{t} d\bar{s} \\
\longrightarrow& 0,
\end{align*}
as $\kappa$ goes to 0. This shows that Expression \ref{e.ca.1a} goes to 0. The proof for Expression \ref{e.ca.1d} is similar.

To prove that Expression \ref{e.ca.1c} goes to 0, similar to the above argument, it suffices to show that
\begin{align*}
\tilde{C}_2 := \int_{I^4}d\hat{t} d\hat{s}\ J_{ij}(\hat{t})& \varrho^{v,\prime}_{j,s}
\varrho^{\bar{v},\prime}_{i,\bar{s}}\\
\times &\ \kappa^2 e^{-\kappa^2|\sigma_i(\hat{t}) - \varrho_{i,s}^v|^2/8}e^{-\kappa^2|\sigma_j(\hat{t}) - \varrho_{j,s}^v|^2/8}
e^{-\kappa^2|\sigma_i(\hat{t}) - \varrho_{i,\bar{s}}^{\bar{v}}|^2/8}e^{-\kappa^2|\sigma_j(\hat{t}) - \varrho_{j,\bar{s}}^{\bar{v}}|^2/8}
\end{align*}
goes to 0 as $\kappa$ goes to infinity.

Let $\omega : \hat{t} \equiv (t, \bar{t}) \mapsto (\sigma_i(\hat{t}), \sigma_j(\hat{t}))$. And let $\hat{S} = \omega(I^2) \subset \bR^2$. A simple change of variable argument and apply Item 2 in the lemma in the appendix of \cite{EH-Lim02}, or using Item 4 in Lemma 5 in \cite{CS-Lim03}, will yield
\begin{align*}
|\tilde{C}_2| \leq& \kappa^2 \int_{I^2} \int_{\hat{S}} e^{-\kappa^2|\omega - \hat{\varrho}_k^{v}(s)|^2/8}e^{-\kappa^2|\omega - \hat{\varrho}_k^{\bar{v}}(\bar{s})|^2/8} d\omega \ \left| \varrho^{v,\prime}_{j,s}
\varrho^{\bar{v},\prime}_{i,\bar{s}} \right| ds d\bar{s} \\
\leq& 4\pi \frac{\kappa^2}{4\pi} \int_{I^2} \int_{\bR^2} e^{-\kappa^2|\omega - \hat{\varrho}_k^{v}(s)|^2/8}e^{-\kappa^2|\omega - \hat{\varrho}_k^{\bar{v}}(\bar{s})|^2/8} d\omega \ \left| \varrho^{v,\prime}_{j,s}
\varrho^{\bar{v},\prime}_{i,\bar{s}} \right| ds d\bar{s} \\
\leq& 4\pi\int_{I^2} e^{-\kappa^2|\hat{\varrho}_k^{v}(s) - \hat{\varrho}_k^{\bar{v}}(\bar{s})|^2/16}|\varrho^{v,\prime}_{j,s}
\varrho^{\bar{v},\prime}_{i,\bar{s}}| ds d\bar{s},
\end{align*}
which converges to 0 as $\kappa$ goes to infinity.

To prove that Expression \ref{e.ca.1b} goes to 0, similar to the above argument, it suffices to show that
\begin{align*}
\tilde{C}_3 := \int_{I^4}d\hat{t} d\hat{s}\ J_{ij}(\hat{t})& \varrho^{v,\prime}_{k,s}
\varrho^{\bar{v},\prime}_{k,\bar{s}}\\
\times&\ \kappa^2 e^{-\kappa^2|\sigma_i(\hat{t}) - \varrho_{i,s}^v|^2/8}e^{-\kappa^2|\sigma_k(\hat{t}) - \varrho_{k,s}^v|^2/8}
e^{-\kappa^2|\sigma_j(\hat{t}) - \varrho_{j,\bar{s}}^{\bar{v}}|^2/8}e^{-\kappa^2|\sigma_k(\hat{t}) - \varrho_{k,\bar{s}}^{\bar{v}}|^2/8}
\end{align*}
goes to 0 as $\kappa$ goes to infinity.

Let $\omega: \hat{s} \equiv (s,\bar{s}) \mapsto y = (y_+, y_-)$ with $s \mapsto y_+ = \varrho_{k,s}^v$, $\bar{s} \mapsto y_- = \varrho_{k,\bar{s}}^{\bar{v}}$.
And define $x_+: y_+ \mapsto \varrho_{i,s}^v$, $x_-: y_- \mapsto \varrho_{j,\bar{s}}^{\bar{v}}$. Let $x(\omega) = (x_+(\omega), x_-(\omega))$.

After a change of variables, the above expression becomes \beq \tilde{C}_3 = \int_{I^2}\int_{\omega(I^2)} \kappa^2 e^{-\kappa^2|\hat{\sigma}_k(\hat{t}) - x(\omega)|^2/8} e^{-\kappa^2|\omega - (\sigma_k(\hat{t}), \sigma_k(\hat{t}))|^2/8}J_{ij}(\hat{t})\ d\omega d\hat{t} . \nonumber \eeq

Let $D \subset \bR^2$ such that $D = \{(s, s): s\in \bR\}$ and let $D_\epsilon$ be a diagonal strip, center at $D$, with width $2\epsilon$. Think of a vertical strip center at the $y$-axis, with width $2\epsilon$. Then rotate this strip clockwise 45 degrees, and we will obtain $D_\epsilon$.

Since $\omega(I^2)$ is bounded, we see that $D_\epsilon \cap \omega(I^2)$ will have area that tends to 0 as $\epsilon$ goes to 0. And for any $\omega \in D_\epsilon^c$, we have $|\omega - (s,s)| \geq \epsilon$ for any $(s,s) \in D$.

Thus,
\begin{align*}
|\tilde{C}_3| \leq& \int_{I^2}\int_{\omega(I^2)\cap D_\epsilon} \kappa^2 e^{-\kappa^2|\hat{\sigma}_k(\hat{t}) - x(\omega)|^2/8} e^{-\kappa^2|\omega - (\sigma_k(\hat{t}), \sigma_k(\hat{t}))|^2/8}|J_{ij}|(\hat{t})\ d\omega d\hat{t}  \\
&+  \int_{I^2}\int_{\omega(I^2)\cap D_\epsilon^c} \kappa^2 e^{-\kappa^2|\hat{\sigma}_k(\hat{t}) - x(\omega)|^2/8} e^{-\kappa^2|\omega - (\sigma_k(\hat{t}), \sigma_k(\hat{t}))|^2/8}|J_{ij}|(\hat{t})\  d\omega d\hat{t}  \\
\leq& \int_{\omega(I^2)\cap D_\epsilon}\int_{I^2} \kappa^2 e^{-\kappa^2|\hat{\sigma}_k(\hat{t}) - x(\omega)|^2/8} |J_{ij}|(\hat{t})\ d\hat{t} d\omega   \\
&+  \int_{\omega(I^2)\cap D_\epsilon^c}\int_{I^2} \kappa^2 e^{-\kappa^2|\hat{\sigma}_k(\hat{t}) - x(\omega)|^2/8} e^{-\kappa^2\epsilon^2/8}|J_{ij}|(\hat{t})\ d\hat{t} d\omega   \\
\leq& 8\pi\left[ {\rm Area}(\omega(I^2)\cap D_\epsilon) + e^{-\kappa^2\epsilon^2/8} \cdot {\rm Area}(\omega(I^2)) \right].
\end{align*}
Let $\kappa$ goes to infinity, then we have \beq \limsup_{\kappa \rightarrow \infty}|\tilde{C}_3| \leq 8\pi \cdot {\rm Area}(\omega(I^2)\cap D_\epsilon), \nonumber \eeq and this holds for any $\epsilon$. Since the RHS goes to 0 as $\epsilon$ goes to 0, we must have the limit of $\tilde{C}_3$ is 0 as $\kappa$ goes to infinity.
\end{proof}

\begin{thm}(Main Theorem)\label{t.c.1} \\
Let $S$ be any closed and bounded orientable surface in $\bR \times \bR^3$ disjoint from a time-like geometric hyperlink $\uL$. Assume that $\pi_a(\uL)$ intersect $\pi_a(S)$ at finitely many points.

Using Expression \ref{e.ehc.1}, we can now define an operator $\hat{F}_S$,
\begin{align*}
\hat{F}_S[Z(q; \chi(\oL, \uL))] :=& -i\sqrt{4\pi}\ {\rm lk}(\uL, S) \otimes \left[ \sum_{i=1}^3 \hat{E}^{0i} - \sum_{j=1}^3\hat{E}^{\tau(j)} \right]Z(q; \chi(\oL, \uL)) \\
\equiv &  -i\sqrt{4\pi}\ {\rm lk}(\uL, S) Z(q; \chi(\oL, \uL)) \otimes \left[ \sum_{i=1}^3 \hat{E}^{0i} - \sum_{j=1}^3\hat{E}^{\tau(j)} \right].
\end{align*}
Note that the Wilson Loop observable of the colored hyperlink $\chi(\oL, \uL)$, $Z(q; \chi(\oL, \uL)) \in \bR$ was defined by Equation (\ref{e.sux.3}).
\end{thm}

\begin{proof}
This follows applying Corollaries \ref{c.c.3} and \ref{c.c.2} and Lemma \ref{l.c.1} to Expression \ref{e.sux.4}.
\end{proof}

\begin{rem}
\begin{enumerate}
  \item We can write the operator as \beq \hat{F}_S := -i\sqrt{4\pi}\ {\rm lk}(\cdot, S)  \otimes \left[ \sum_{i=1}^3 \hat{E}^{0i} - \sum_{j=1}^3\hat{E}^{\tau(j)} \right], \nonumber \eeq which acts on the Wilson Loop observable of the colored hyperlink $\chi(\oL, \uL)$.
  \item From Theorem \ref{t.c.1}, the operator $\hat{F}_S$ is not a two form with values in $\Lambda^2(V)$, but is a $\mathfrak{su}(2) \times \mathfrak{su}(2)$-valued operator, which depends on a surface $S$.
\end{enumerate}
\end{rem}

Recall that $F_S(\omega)$ is the integral of the curvature over a surface $S$. By averaging all possible spin connections $\omega$ in $\bR^4$, Expression \ref{e.ehc.1} quantizes the total curvature of a surface $S$, into an operator $\hat{F}_S$, which takes only discrete values because the linking number between a surface $S$ and the geometric hyperlink $\uL$ is discrete.

\section{Application to Quantum Cosmology}

Thiemann described in \cite{Thiemann:2002nj}, the problems General Relativity faced. Solving Einstein's Equations will lead to singularities, hence predicting the existence of black holes. The Friedmann-Robertson-Walker (FRW) model, which describes the universe, also faces certain problems.

In this model, the metric is given by \beq ds^2 = -dt^2 +R(t)^2\left[\frac{dr^2}{1-kr^2} + r^2d\Omega_2^2 \right]. \nonumber \eeq The constant $k$ will determine if we have closed, flat or open universe. It will be closed, flat or open, if $k = 1, 0, -1$ respectively. We will focus on $R(t)$, which depends on time. Einstein's Equations predict that $\lim_{t \rightarrow 0^+}R(t) = 0$, resulting in the metric being singular, giving rise to the idea of the classical Big Bang singularity. Now curvature is actually proportional to $1/R(t)^2$, so indeed this singularity is a curvature singularity.

One way to resolve this singularity is to replace it with a `quantum bounce', as introduced by the authors in \cite{PhysRevD.77.024046}. They argued that a FRW model is only valid till very early times. As curvature increases even further, loop quantum gravity takes over and creates an effective repulsive force, overwhelms classical gravitational attraction, and causes a `quantum bounce'. Their theory effectively replaces the big bang which occurs at the start of time, hence resolving the singularity as mentioned above. See \cite{PhysRevD.73.124038}. This is one of the major advancement in Loop Quantum Gravity since 2003. We wish to highlight that in the articles \cite{PhysRevLett.96.141301} and \cite{PhysRevD.73.124038}, the authors seemed to be working in a 3-dimensional manifold, instead of a 4-dimensional manifold, just as quantum gravity theorists who used spin networks also worked in a 3-dimensional space. For a good review of spin networks, we refer the reader to \cite{Baez:1999sr}, \cite{Baez:1999:Online}, \cite{Baez1996253} and \cite{Rovelli:1995ac}.

We can also resolve this singularity using the quantized curvature we obtained earlier. In the FRW model, curvature will be continuous with respect to time. Fix a surface $S \subset \bR^3$. If we choose $\omega$ to be the Levi-Civita connection associated with the metric in the FRW model, then total curvature $F_S(\omega)$ given in Equation (\ref{e.r.2}) will vary continuously with respect to time $t$, and approach infinity as $t \rightarrow 0^+$. If we deform the surface $S$ continuously, then $F_S(\omega)$ will vary continuously.

Quantum gravity, on the other hand, paints us a different picture. As time decreases, a highly curved universe will imply that the space is getting smaller and we are approaching the regime of quantum gravity. When quantum gravity takes over, the total curvature is actually quantized and the quantized operator $\hat{F}_S$ takes discrete values for different surfaces $S$, as the linking number between a surface $S$ and the geometric hyperlink is discrete. Thus, we cannot have continuum values of total curvature, hence the metric given above no longer holds in short distances or high energy and the FRW model is no longer valid at distances and energies in the quantum gravity regime.



\nocite{*}



\end{document}